\documentclass[preprint,12pt]{elsarticle}

\usepackage[utf8]{inputenc}
\usepackage[T1]{fontenc}
\usepackage{amsmath, amsfonts, amssymb}
\usepackage[amsmath,thmmarks]{ntheorem}

\usepackage[initials,nobysame,short-months]{amsrefs}

\theoremstyle{plain}
\theoremheaderfont{\itshape\bfseries}
\theorembodyfont{\normalfont\itshape}
\theoremseparator{:}
\theoremsymbol{}
\newtheorem{theorem}{Theorem}
\newtheorem{proposition}[theorem]{Proposition}

\newtheorem{corollary}[theorem]{Corollary}

\newtheorem{definition}{Definition}
\newtheorem{axiom}{Axiom}

\theoremstyle{nonumberplain}

\theoremstyle{plain}
\theoremheaderfont{\itshape}
\theorembodyfont{\normalfont}
\theoremseparator{:}
\theoremsymbol{}
\newtheorem{remark}{Remark}

\theoremstyle{nonumberplain}
\theoremsymbol{\rule{1ex}{1ex}}
\newtheorem{proof}{Proof}

\def\D{ \mathbb{D} }							% unit disk
\def\dD{ \partial\mathbb{D} }			% unit circle
\def\DD{ \overline{\mathbb{D}} }	% unit disk´
\def\d{ \mathrm{d} }							% integral d
\def\I{ \mathrm{i} }							% imaginary unit
\def\E{ \mathrm{e} }							% e
\def\Energy{ \mathrm{E} }					% Energy
\def\NN{ \mathbb{N} }													% positive integers
\def\RN{ \mathbb{R} }													% real numbers
\def\QN{ \mathbb{Q} }													% rational numbers
\def\CN{ \mathbb{C} }													% complex numbers
\def\RNc{ \Delta_{1} }												% computable real numbers
\def\Csemi{ \mathbf{C}_{1} }									% semi computable numbers
\def\Cweak{ \mathbf{C}_{2} }									% weakly computable numbers
\def\B{ \mathcal{B} }													% Banach space
\def\Bc{ \mathcal{B}_{\mathrm{c}} }					  % B-computablle functions
\def\C{ \mathcal{C} }								  				% continuous functions
\def\Cc{ \mathcal{C}_{\mathrm{c}} }						% computable continuous functions
\def\H12{ H^{1/2} }														% Sobolev H^{1/2}
\def\bm{ \mathbf{m} }
\def\bsx{ \boldsymbol x}											% bold sequence x
\def\M{ \mathcal{M} }													% Set1 M1 und M2
\def\Pd{ \Pi_{1} }														% Pi upper 
\def\Pu{ \Sigma_{1} }													% Pi lower
\newcommand{\Op}[1]{\mathrm{#1}}               % allgemeine Operatoren
\DeclareMathOperator*{\grad}{grad}			       % gradient
\begin{document}

\begin{frontmatter}

\title{Arithmetic Complexity of Solutions of the Dirichlet Problem}

\author[affilPohl,affilBoche]{Holger Boche}\ead{boche@tum.de}
\author[affilPohl]{Volker Pohl}\ead{volker.pohl@tum.de}
\author[affilPoor]{H. Vincent Poor}\ead{poor@princeton.edu}

%% Author affiliationBMBF Research Hub 6G-life, M\"unchen, Germany.
\affiliation[affilPohl]{organization={Technische Universit\"at M\"unchen},
            addressline={Arcisstr. 21}, 
            city={M\"unchen},
            postcode={80333}, 
            %state={},
            country={Germany}}
						
\affiliation[affilBoche]{organization={BMBF Research Hub 6G-life},
            %addressline={}, 
            city={M\"unchen},
            postcode={80333}, 
            %state={},
            country={Germany}}

\affiliation[affilPoor]{organization={Department of Electrical and Computer Engineering, Princeton University},%Department and Organization
            %addressline={}, 
            city={Princeton},
            postcode={08544}, 
            state={NJ},
            country={USA}}

%\cortext[cor1]{Corresponding author}
\tnotetext[label1]{This work was partly supported by the German Federal Ministry of Research, Technology and Space
(BMFTR) within the program "Souverän. Digital. Vernetzt” in the joint project 6G-life under Grant 16KISK002.
The work of H.~V.~Poor was partly supported by the U.S National Science Foundation under Grants CNS-2128448 and ECCS-2335876.}

% ========== Abstract ==========
\begin{abstract}
The classical Dirichlet problem on the unit disk can be solved by different numerical approaches. 
The two most common and popular approaches are the integration of the associated Poisson integral and, by applying Dirichlet's principle, solving a particular minimization problem.
For practical use, these procedures need to be implemented on concrete computing platforms.
This paper studies the realization of these procedures on Turing machines, the fundamental model for any digital computer.
We show that on this computing platform both approaches to solve Dirichlet's problem yield generally a solution that is not Turing computable, even if the boundary function is computable.
Then the paper provides a precise characterization of this non-computability in terms of the Zheng--Weihrauch hierarchy. 
For both approaches, we derive a lower and an upper bound on the degree of non-computability in the Zheng--Weihrauch hierarchy.
\end{abstract}

%% Keywords
\begin{keyword}
Harmonic Analysis \sep Computability \sep Boundary value problems \sep Dirichlet principle \sep Optimization \sep Zheng--Weihrauch hierarchy

%% PACS codes here, in the form: \PACS code \sep code

%% MSC codes here, in the form: \MSC code \sep code
%% or \MSC[2008] code \sep code (2000 is the default)
\end{keyword}

\end{frontmatter}

\newpage
% =====================================================================================================
% ========== INDRODUCTION =============================================================================
% =====================================================================================================
\section{Introduction}
\label{sec:intro}

Dirichlet problems play a central role in numerous areas of engineering and applied mathematics.
These problems search for functions that satisfy a particular partial differential equations in a certain region and that take prescribed values on the boundary of this region.
The most simple case is probably the Dirichlet problem that asks for solutions of Laplace's equation on the unit disk with prescribed values on the unit circle (cf. Section~\ref{sec:ProblemStatment} for a precise formulation).   
An important tool in solving this problem is \emph{Dirichlet's principle}.
It states that the solution of the Dirichlet problem can be found by variational methods \cite{Funk_Variationsrechnung_62} based on the minimization of the \emph{Dirichlet energy}.
Since the Dirichlet energy is equivalent to the norm in the Sobolev space $W^{1,2}_{0}$, this approach reduces to a convex optimization problem for which there exist several numerical approaches \cites{Bertsekas_ConvexOpt,Boyd_ConvexOpt}.
Moreover, Dirichlet's principle is one cornerstone in the development of modern computing methods based on calculus of variations \cites{GanderWanner_SIAM12,Repin_CMAM17} such as Ritz's method \cites{Ritz_1909,Leipholz_87}, the Galerkin method \cites{Galerkin_1915,Grote_GalerkinFEM_06}, the finite element method \cites{Courant_BullAMS43,BrennerScott_FEM2008,Acosta_FEM_SIAM17}, or the boundary element method \cite{CHENG_BME2005}, to mention only a few.
Apart from optimization techniques, the solution of Dirichlet's problem may also be computed based on the Poisson integral.

In principle there exist many different approaches to solve the Dirichlet problem, either using optimization technique based on the Dirichlet principle or by numerically integration based on the Poisson integral.
Both approaches are active research field in harmonic and Fourier analysis and it is an important practical question, whether these theoretical approaches can be implemented as concrete algorithms on a certain computing platform.
Thereby, the notion of an "algorithm" strongly depends on the underlying computing platform.
Thus, it might be possible to implement a certain mathematical approach on one computing platform, whereas it might be impossible to implement this approach on another computing platform.
Nowadays there exist, in principle, several different computing platforms, e.g. digital, quantum, or neuromorphic computing.
Among these different platforms, digital computing is by far the most widely used platform and the corresponding theory of algorithms for this platform (i.e. Turing's computing theory) is the most advanced computing theory do date.

In this paper, we  investigate whether it is possible to implement the two mentioned approaches for solving the Dirichlet problem on the unit disk (namely the approach using Poisson integral and the variational approach based on Dirichlet's principle) on a Turing machine. Thus, we ask whether there exist algorithms on Turing machines for these two approaches that are able to compute the corresponding solution.
Since Turing machines can only compute with rational numbers, any algorithm on a Turing machine can generally only compute a (rational) approximation of the true solution.
It is also clear that any Turing machine has to stop after finitely many steps.
Consequently, and according to Turing's definition of computability, any algorithm on a Turing machine has to stop the approximation process of the true solution after finitely many operations if a required approximation error is achieved.
Only if such a controlled termination, depended on the approximation error, is possible, one says that the algorithm \emph{effectively converges} to the solution and that the solution is \emph{Turing computable}.   
Such an effective error control is particularly important for applications with stringent requirements on security and for autonomous systems that need to guarantee error bounds autonomously without any human support.

This paper investigates whether it is possible to find effective algorithms to solve the Dirichlet problem either based on minimization techniques or based on the computation of the Poisson integral on Turing machines.
This will be done for two natural sets of boundary functions:
1) The set of computable continuous functions of finite Dirichlet energy.
2) The set of continuous functions that can effectively approximated by a harmonic series in the (Dirichlet) energy norm.
It is shown that there exists no algorithm on a Turing machine that is able to effectively solve the minimization problem associated with the Dirichlet problem for all boundary functions in this first set and it is shown that it is not possible to effectively compute the Poisson integral (in order to solve the Dirichlet problem) for all boundary functions in the second set. So in both cases, the solution of the Dirichlet problem is generally not Turing computable.
Similarly as the Turing degrees characterize the undecidability of certain problems, the Zheng--Weihrauch hierarchy provides a possibility to characterize the degree of Turing non-computability.
Based on the Zheng--Weihrauch hierarchy we will derive precise lower and upper bounds on the degree of non-computability of the solutions of the Dirichlet problem in the two cases.

The problems investigate in this paper also have a quite general theoretical aspect.
In science, one usually has general physical theories, e.g. Maxwell's equations that describe the phenomenons of classical electrodynamics and electrostatics.
These theories are usually based on some general physical principles with a corresponding precise mathematical formulation.
For example, the \emph{principle of minimum energy}, i.e. Dirichlet's principle \cites{Ritz_1909, Mikhlin_VariationalMethods}, leads to variational calculus.
Similarly, the \emph{maximum principle} leads to the solution of the Dirichlet problem based on the Poisson integral.
Then it is an interesting question whether these physical theories are compatible with a certain computing model in the following sense:
Assume that a problem in such a physical theory has a natural defined computable input (with respect to the computing model). 
Is it then true that also the physical performance indicator of this problem (e.g. the minimal energy or the maximum value) is computable in this computing model?
The results of this paper will show that for digital computing, i.e. for Turing machines, the answer is generally negative.

The remainder of this paper is organized in the following way:
Section~\ref{sec:Notation} clarifies general notation, Section~\ref{sec:Computability} recalls basic concepts from computability analysis, and
Section~\ref{sec:ZWHierarchy} gives a short introduction to the Zheng--Weihrauch Hierarchy of non-computability. 
After these preparations, Section~\ref{sec:ProblemStatment} gives a precise formulation of the Dirichlet problem on $\D$ and provides a detailed problem statement.
Then Section~\ref{sec:MainMinimizationProblem} investigates the computability of solutions of the Dirichlet problem bases on minimizing Dirichlet's energy, and Section~\ref{sec:MainDirichletIntegral} investigates the computability of solutions obtain by integrating Poisson integral.
The paper closes with a short summary in Section~\ref{sec:Conclusion}.

% ===============================================
% ========== Dirichlet Principle ================
% ===============================================
\section{General Notation and Function Spaces}
\label{sec:Notation}

We use standard notations from harmonic analysis and functional analysis (see, e.g. \cites{Liflyand_HarmonicAnalysis,Rudin,Zygmund}), in particular stands $\D = \left\{ z\in\CN : \left|z\right| < 1\right\}$ and $\dD = \left\{ z\in\CN : \left|z\right| = 1\right\}$ for the \emph{unit disk} and the \emph{unit circle}, respectively, in the complex plane $\CN$.
Then we write $\DD = \D \cup \dD$ for the \emph{closed unit circle} in $\CN$.
For any $1\leq p \leq \infty$, $L^{p}(\dD)$ denotes the usual space of $p$-integrable functions on $\dD$.
The Banach space of continuous functions on $\dD$ equipped with the uniform norm $\left\|f\right\|_{\infty} = \max_{\zeta\in\dD} f(\zeta)$ is denoted by $\C(\dD)$.

For any $f\in L^{1}(\dD)$, its \emph{Fourier coefficients} are given by
\begin{equation}
\label{equ:FourierCoef}
\begin{aligned}
	a_{n}(f) &= \frac{1}{\pi}\int^{\pi}_{-\pi} f(\E^{\I\theta})\, \cos(n \theta)\, \d \theta\,,\quad n=0,1,2,\dots\\
	b_{n}(f) &= \frac{1}{\pi}\int^{\pi}_{-\pi} f(\E^{\I\theta})\, \sin(n \theta)\, \d \theta\,,\quad n=1,2,3,\dots\,,
\end{aligned}
\end{equation}
and any $f\in L^{2}(\dD) \subset L^{1}(\dD)$ can be recovered from its Fourier coefficients by the harmonic series
\begin{equation}
\label{equ:FourierSeries}
	f(\E^{\I\theta})
	= \frac{a_{0}(f)}{2} + \sum^{\infty}_{n=1} \left[ a_{n}(f)\cos(n\theta) + b_{n}(f)\sin(n\theta) \right]\,,
\end{equation}
where the sum converges in the norm of $L^{2}(\dD)$ and equality in \eqref{equ:FourierSeries} means equality in $L^{2}(\dD)$.
According to Parseval's identity, one has
\begin{equation*}
	\left\|f\right\|^{2}_{2}
	= \frac{\left|a_{0}(f)\right|^{2}}{4} + \frac{1}{2}\sum^{\infty}_{n=0} \left( \left|a_{n}(f)\right|^{2} + \left|b_{n}(f)\right|^{2} \right)\,.
\end{equation*}
Often, especially in signal processing, this value is called the energy of $f$ and so $L^{2}(\dD)$ is said to be the space of functions on $\dD$ with finite energy.
Nevertheless, from a physical point of view the so-called \emph{Dirichlet energy} \cites{Ross_DirichletSpace,Arcozzi_DirichletSpace11} is often of much more physical importance.
For $f \in L^{2}(\dD)$, it is given by
\begin{equation*}
	\Energy(f) 
	= \frac{1}{2}\sum^{\infty}_{n=1} n \left( \left|a_{n}(f)\right|^{2} + \left|b_{n}(f)\right|^{2} \right)\;,
\end{equation*}
and we write
\begin{equation*}
	\H12(\dD) = \left\{ f \in L^{2}(\dD)\ :\ \Energy(f) < \infty \right\}
\end{equation*}	
for the set of all $L^{2}(\dD)$-functions of finite Dirichlet energy. 
Equipped with the norm
\begin{equation}
\label{equ:H12_norm}
	\left\|f\right\|_{H^{1/2}}
	= \left( \tfrac{\left|a_{0}(f)\right|^{2}}{4} + \Energy(f) \right)^{1/2}\,,
\end{equation}
$\H12(\dD)$ becomes a complete Banach space.
In fact $\H12(\dD) = W^{\frac{1}{2},2}(\dD)$ is the usual Sobolev-space with index $1/2$ \cites{Adams_SobolevSpaces,Evans_PDEs}.

A real valued function $u : \D \to\RN$ is said to be \emph{harmonic} in $\D$ 
if, at every point $z = x+\I y \in\D$, it is twice partially differentiable with respect to $x$ and $y$
and satisfies 
\begin{equation*}
	\left(\Delta u\right)(z) = \frac{\partial^{2} u}{\partial x^{2}}(z) + \frac{\partial^{2} u}{\partial y^{2}}(z) = 0\,,
	\quad\text{for all}\ z\in\D\,.
\end{equation*}
With every function $u$ that is harmonic in $\D$, we associate its so-called \emph{Dirichlet integral}
\begin{equation}
\label{equ:DirichletEnergy}
	\Energy(u) = \frac{1}{2\pi}\iint_{\D} \left\|(\grad u)(z)\right\|^{2}_{\RN^{2}}\, \d x \d y\;,
\end{equation}
and the value $\Energy(u)$ is said to be the \emph{Dirichlet energy} of $u$ \cite{Ross_DirichletSpace}.

% =========================================
% ========== Computability ================
% =========================================
\section{Computability Analysis}
\label{sec:Computability}

The theoretical foundations for analyzing digital computation was laid by \emph{Alain Turing} \cites{Turing_1937,Turing_1938}.
Its theory is based on the concept of a \emph{Turing machine}, i.e. an abstract machine that is able to manipulate symbols on a tape according to some simple rules, and any algorithm for a digital computer can be translated into a corresponding program on a Turning machine.
However, in contrast to real-world computers, a Turing machine has no limitations on storage, energy consumption, or computation time. Moreover, all computations are error free.
This way, the model of a Turing machine provides a fundamental upper bound on the capability of any digital computers.

The most fundamental limitation of a Turing machine (i.e. of any digital computer) is the property that it can only compute with rational numbers.
This entails that the inputs and outputs of a Turing machine can only be rational numbers.
So the question arises how is it possible to compute non-rational quantities and how to have non-rational quantities as inputs?
Moreover, how it is possible to have more complicated objects like sequences or functions as in- or outputs of a Turing machine?
In other words, what kind of objects can actually be handled by a Turing machine?
The answer is simple: Turing machines can only handle algorithms for Turing machines.
Therefore all mathematical objects needed to be described by an algorithm on a Turing machine. Then if we want to give this object to a Turing machine, we simply hand over the corresponding algorithm. 
It turns out, that not all mathematical object can be described by an algorithm on a Turing machine. 
Those, for which such a program exist, are said to be \emph{(Turing-) computable}.
These ideas are made precise by \emph{computability analysis} \cites{PourEl_Computability,Weihrauch_ComputableAnalysis}. It studies those parts of analysis and functional analysis that can be carried out on Turing machines.
In the first subsection, we briefly recall some of its notions and concepts. Thereby, we only assume that the reader has some basic knowledge about \emph{recursive functions}. 
These are functions $\NN^{N} \to \NN$, for some $N\in\NN$, that are build from elementary computable functions and recursions.
For this paper, it is only important that these are functions with integer inputs and outputs that can be implemented on a Turing machine.

Every real number $x\in\RN$ is the limit of a sequence $\left\{ r_{n} \right\}_{n\in\NN} \subset\QN$ of rational numbers .
For $x$ to be computable, two additional conditions need to be satisfied: 
\begin{enumerate}
\item the \emph{sequence} $\left\{ r_{n} \right\}_{n\in\NN}$ needs to be \emph{computable}
\item the \emph{convergence} $r_{n} \to x$ needs to be \emph{effective}
\end{enumerate}
We first give a formal definition of both properties:

\begin{definition}[Computable rational sequence]
A sequence $\left\{ r_{n} \right\}_{n\in\NN} \subset\QN$ of rational numbers is said to be computable if there exist recursive functions $a,b,s : \NN\to\NN$ with $b(n) \neq 0$ for all $n\in\NN$ such that
\begin{equation*}
	r_{n} = (-1)^{s(n)} \frac{a(n)}{b(n)}\,,
	\qquad\text{for all}\ n\in\NN\;.
\end{equation*}
\end{definition}

\begin{definition}[Effective convergence]
A sequence $\left\{ x_{n} \right\}_{n\in\NN} \subset \RN$ is said to \emph{effectively converge} to $x \in\RN$ if there there exists a recursive function $e : \NN\to\NN$ such that
\begin{equation}
\label{equ:EffctConv}	
	n \geq e(M)
	\qquad\text{implies}\quad
	\left| x - x_{n} \right|	< 2^{-M}\,.
\end{equation}
\end{definition}
So effective convergence not only means that $x_{n}$ converges to $x$ but also that one is able to control the approximation error in the sense that to every $M\in\NN$ one can determine an index $N_{0} = e(M)$ such that the approximation error $\left| x - x_{n} \right|$ is guaranteed to be less than $2^{-M}$ provided $n\geq N_{0}$.

\begin{definition}[Computable number]
\label{def:CompNumber}
An $x\in\RN$ is said to be \emph{computable} if there exists a computable sequence of rational numbers $\left\{ r_{n} \right\}_{n\in\NN} \subset\QN$ 
that effectively converges to $x$.
\end{definition}

\begin{remark}
For reasons that become will clear in the next section, we denote the set of all computable real numbers by $\RNc$.
\end{remark}

The idea of effective convergence can be extended to other objects like sequences, functions, etc.:

\begin{definition}[Computable sequence]
A sequence $\left\{ x_{n} \right\}_{n\in\NN}\subset\RN$ of real numbers is said to be \emph{computable} if there exists a double sequence $\left\{ r_{n,m} \right\}_{n,m\in\NN} \subset \QN$ of rational numbers that converges uniformly and effectively to $\left\{ x_{n} \right\}_{n\in\NN}$, i.e.
there exists a recursive function $e: \NN\times \NN \to \NN$ such that for all $n,M \in \NN$ 
\begin{equation*}
	m\geq e(n,M)
	\quad\text{implies}\quad
	\left|x_{n} - r_{n,m} \right| < 2^{-M}\;.
\end{equation*}
\end{definition}

There are many different notions that describe the computability of functions.
The weakest and most natural of these notions is the so called Banach--Mazur computability.

\begin{definition}[Banach--Mazur computable function]
\label{def:BMCompFunc}
A real function $f$ on $\dD$ is said to be \emph{Banach--Mazur computable} (or \emph{sequentially computable}) if $f$ maps every computable sequences of computable numbers $\left\{ \theta_{n} \right\}_{n\in\NN}$ into a computable sequence $\left\{ f(\E^{\I\theta_{n}}) \right\}_{n\in\NN}$ of computable numbers.
\end{definition}

This paper mainly considers Banach spaces $\B$ of functions on the unit circle.
Apart from the pointwise computability of such functions, as characterized by the Banach--Mazur computability, one also needs a "computability structure" that characterizes the effective convergence in the corresponding norm.
Such a computability structure can be defined in an axiomatic way by characterizing the set $\mathcal{S}$ of all "computable sequences" in $\B$ \cite{PourEl_Computability}:

\begin{axiom}
Let $\left\{ f_{n}\right\}_{n\in\NN}$ and $\left\{ g_{n}\right\}_{n\in\NN}$ be computable sequences in $\B$, let $\left\{\alpha_{n,k}\right\}_{n,k\in\NN}$ $\left\{\beta_{n,k}\right\}_{n,k\in\NN}$ be computable double sequence of real numbers, and let $d : \NN \to\NN$ be a recursive function, then the sequence $\left\{ h_{n} \right\}_{n\in\NN}$ defined by
\begin{equation*}
	h_{n} = \sum^{d(n)}_{k=0} \left( \alpha_{n,k}\, f_{k} + \beta_{n,k}\, g_{k} \right)\,,\qquad n\in\NN
\end{equation*}
is a computable sequence in $\B$.
\end{axiom}

\begin{axiom}
Let $\left\{ f_{n,k} \right\}_{n,k\in\NN}$ be a computable double sequence in $\B$ and let $\left\{ f_{n} \right\}_{n\in\NN}$ be a sequence in $\B$ such that there exists a recursive function $e : \NN\to\NN$ such that
\begin{equation}
\label{equ:EffConvInKN}
	k \geq e(n,N)
	\quad\text{implies}\quad
	\left\|f_{n,k} - f_{n} \right\|_{\B} \leq 2^{-N}\,.
\end{equation}
Then $\left\{ f_{n} \right\}_{n\in\NN}$ is a computable sequence in $\B$.
\end{axiom}

\begin{axiom}
If  $\left\{ f_{n} \right\}_{n\in\NN}$ is a computable sequence in $\B$, then $\left\{ \left\|f_{n}\right\|_{\B} \right\}_{n\in\NN}$ is a computable sequence of real numbers.
\end{axiom}

\begin{remark}
In these axioms we used that a double sequence $\left\{ \alpha_{n,k} \right\}_{n,k\in\NN}$ is said to be \emph{computable} if it is mapped on a computable sequence by one of the standard recursive pairing functions $\NN\times\NN \to \NN$ \cite{PourEl_Computability}.
\end{remark}

\begin{remark}
If \eqref{equ:EffConvInKN} is satisfied then one says that $\left\{ f_{n,k} \right\}_{n,k\in\NN}$ converges to $\left\{ f_{n} \right\}_{n\in\NN}$ as $k\to\infty$, effectively in $k$ and $n$.
\end{remark}
Having defined a computable structure on a Banach space $\B$ (i.e. the set $\mathcal{S}$ of all sequence that satisfy the previous axioms), an element $f \in \B$ is said to be computable, if the sequence $\left\{ f,f,f, \dots \right\}$ is a computable sequence in $\B$.

In this paper, we only consider Banach spaces $\B$ with a computability structures $\mathcal{S}$ that is \emph{effectively separable} in the following sense.

\begin{definition}[Effectively separable]
A Banach space $\B$ with a computability structure $\mathcal{S}$, denoted by $(\B,\mathcal{S})$, is called \emph{effectively separable} if there exists a computable sequence $\left\{ e_{n} \right\}_{n\in\NN} \subset \B$ whose linear span is dense in $\B$.
Such a sequence $\left\{ e_{n} \right\}_{n\in\NN}$ is called an \emph{effective generating set} of $(\B,\mathcal{S})$.
\end{definition}

Since we consider functions on $\dD$ (i.e. $2\pi$-periodic functions) the most natural effective generating set for these spaces is the sequence of trigonometric polynomials, i.e. the set
\begin{equation*}
	\left\{ 1 \right\} \cup \left\{ \cos(n\theta) \right\}_{n\in\NN} \cup \left\{ \sin(n\theta) \right\}_{n\in\NN}\,,
	\qquad \theta\in [-\pi,\pi)\;.
\end{equation*}
Then a sequence $\left\{ \varphi_{n} \right\}_{n\in\NN}$ is said to be a \emph{computable sequence of trigonometric polynomials} if the functions $\varphi_{n}$ have the form
\begin{equation}
\label{equ:TrigPoly}
	\varphi_{n}(\E^{\I\theta})
	= \frac{a_{n}(0)}{2} + \sum^{d(n)}_{k=1} \left[ a_{n}(k)\, \cos(k \theta) + b_{n}(k)\, \sin(k \theta) \right]\,,
	\quad \theta\in [-\pi,\pi)\,,
\end{equation}
wherein $\left\{ a_{n}(k) \right\}_{n,k\in\NN}, \left\{ b_{n}(k) \right\}_{n,k\in\NN}$ are computable (double) sequences of computable numbers and $d : \NN\to \NN$ is a recursive function that computes for every $n\in\NN$ the degree $d(n)$ of the trigonometric polynomial $\varphi_{n}$.

For Banach spaces $\B$ of functions on the $\dD$, there exists a generic definition of computability that satisfies the previous axioms of a "computability structure" \cite{PourEl_Computability}.
This definition defines the computable objects in $\B$ as those functions that can effectively be approximated by trigonometric polynomials:

\begin{definition}[Computable functions in Banach spaces]
\label{def:ComputableBanachSpace}
1) A function $f\in\B$ in a Banach space $\B$ of functions on $\dD$ is said to be \emph{$\B$-computable} if there exists a computable sequence $\left\{ \varphi_{n} \right\}_{n\in\NN}$ of trigonometric polynomials that effectively converges to $f$, i.e. there exists a recursive function $e : \NN\to\NN$ such that
\begin{equation*}
	n\geq e(M)
	\quad\text{implies}\quad
	\left\|f - \varphi_{n} \right\|_{\B} < 2^{-M}\,.
\end{equation*}
2) A sequence $\left\{ f_{n} \right\}_{n\in\NN} \subset \B$ is said to be $\B$-computable if there exists a computable double sequence $\left\{ \varphi_{n,k} \right\}_{n,k\in\NN}$ of trigonometric polynomials and a recursive function $e : \NN\to\NN$ such that
\begin{equation*}
	k \geq e(n,N)
	\quad\text{implies}\quad
	\left\|\varphi_{n,k} - f_{n} \right\|_{\B} \leq 2^{-N}\,.
\end{equation*}
The set of all $\B$-computable function is denoted by $\Bc$.
\end{definition}

\begin{remark}
We mainly need the two Banach space $\Cc(\dD)$ and $H^{1/2}_{\mathrm{c}}(\dD)$.
A function $f \in \Cc(\dD)$ is said to be a \emph{computable continuous function}.
\end{remark}

\begin{remark}
It is not hard to verify \cite{PourEl_Computability} that the $\B$-computable sequences satisfy the three axiom of a computable structure defined above and that the set of trigonometric polynomials is an effective generating set for these Banach spaces $\B$ together with the computability structure defined by Definition~\ref{def:ComputableBanachSpace}.
\end{remark}

\begin{remark}
Note that Axiom~$3$ of the computability structure implies that if $\big\{ f_{n} \big\}_{n\in\NN} \subset\B$ is a computable sequence that effectively converges to $f$ then $\big\{ \left\|f_{n}\right\|_{\B} \big\}_{n\in\NN}$ is a computable sequence of computable numbers that effectively converges to $\left\|f\right\|_{\B}$. So in particular, $\left\|f\right\|_{\B}$ is a computable number for every computable $f \in\B$.
\end{remark}

\begin{remark}
We also note that every computable continuous function is also Banach--Mazur computable \cite{PourEl_Computability}. 
Conversely there are Banach--Mazur computable functions that are not computable continuous functions.
\end{remark}

% ======================================================
% ========== Zheng--Weihrauch Hierarchy ================
% ======================================================
\section{The Zheng--Weihrauch Hierarchy of Non-Computability}
\label{sec:ZWHierarchy}

In computer science and mathematical logic, there exist important decision problems that are not algorithmic solvable.
Probably the best known examples is the \emph{halting problem} which is known to be \emph{undecidable} \cites{Church_AJM1936,Turing_1937}, i.e.
there exists no general algorithm that solves the halting problem for all possible inputs. 
Closely related is G\"odels incompleteness theorem \cite{Goedel_1931} that shows that in any consistent system of axioms there always exist true statements that cannot be proven within the system. 
This existent of undecidable problems, inspired the introduction of \emph{Turing degrees}, due to Post and Kleene	\cites{Post_AMS44,KleenePost_54}, that characterize in some sense the degree of unsolvability.

This paper considers functions that are obtained as solution of a certain mathematical operation, namely solutions of the Dirichlet problem. 
The input data of the mathematical operation (i.e. the data of the Dirichlet problem) are assumed to be objects that can be described in an algorithmic way (cf. Section~\ref{sec:Computability}). 
However, as we will see, the functions, i.e. the solutions of the mathematical operation, are generally not algorithmically computable.  
Therefore, we ask for the \emph{degree of non-computability} of these solutions measured in an appropriate hierarchy of algorithmically non-computability.
Since the values of the solution functions are real numbers, we apply the \emph{Zheng--Weihrauch hierachy} \cite{ZhengWeihrauch_MathLog01} which is a natural extension of the theory of Turing degrees for real numbers.
This hierarchy characterizes the non-computability of real numbers.
It consists of sets of real numbers. Numbers in the same set (i.e. in the same layer of the hierarchy) have the same degree of non-computability. 
The first layer of this hierarchy contains the computable numbers. Higher layers contain increasingly less computable numbers.
These layers thus describe the arithmetically complexity of the numbers in this layer. Therefore, we sometimes call these layers \emph{complexity classes}.
In this paper, we only need the first two layers of the Zheng--Weihrauch hierarchy. These two layers are briefly described and motivated in the following two subsections.
For completeness and to place our results in larger context, the definition of the whole hierarchy is given in \ref{sec:AppZW}.

% ===== The first Layer of Zheng-Weihrauch =====
\subsection{The first layer in the Zheng--Weihrauch hierarchy}

Following \cite{ZhengWeihrauch_MathLog01}, we start by defining two sets of \emph{semi-computable} numbers:

\begin{definition}[The sets $\Sigma_{1}$, $\Pi_{1}$ and $\Csemi$]
An $x \in\RN$ belongs to $\Pu$ if and only if there exists a computable sequence $\left\{ \zeta_{n} \right\}_{n\in\NN} \subset\RNc$ such that
\begin{equation*}
	\zeta_{n+1} \geq \zeta_{n}\,,\ \text{for all}\ n\in\NN\,,
	\quad\text{and}\quad
	\lim_{n\to\infty} \zeta_{n} = x\;.
\end{equation*}
An $x \in\RN$ belongs to $\Pd$ if and only if there exists a computable sequence $\left\{ \xi_{n} \right\}_{n\in\NN}\subset\RNc$ such that
\begin{equation}
\label{equ:propXi}
	\xi_{n+1} \leq \xi_{n}\,,\ \text{for all}\ n\in\NN\,,
	\quad\text{and}\quad
	\lim_{n\to\infty} \xi_{n} = x\;.
\end{equation}
We write $\Csemi = \Sigma_{1} \cup \Pi_{1}$ for the set of all semi-computable numbers.
\end{definition}

\begin{remark}
For obvious reasons, $x\in\Sigma_{1}$ is usually called \emph{left-computable}, whereas $x\in\Pi_{1}$ is said to be \emph{right-computable}.
\end{remark}
A semi-computable number is generally not computable. However, it is not hard to see \cite{PourEl_Computability} that if $x\in\RN$ is left- \emph{and} right-computable, then it is also computable in the sense of Def.~\ref{def:CompNumber}.

\begin{proposition}
\label{prop:Delta1isUnion}
A number $x \in\RN$ belongs to $\Delta_{1}$ if and only if $x \in \Sigma_{1} \cap \Pi_{1}$, that is $\Delta_{1} = \Sigma_{1} \cap \Pi_{1}$.
\end{proposition}

\begin{remark}
Note that $\Delta_{1}$ is a proper subset of $\Sigma_{1}$ and $\Pi_{1}$, i.e. there exist numbers in $\Sigma_{1}$ and $\Pi_{1}$ that are not computable.
\end{remark}

It can be shown \cite{ZhengWeihrauch_MathLog01} that $\Delta_{1}$ is closed under the usual arithmetical operations addition, subtraction, multiplication and division.
This is not true for the sets $\Sigma_{1}, \Pi_{1}, \Csemi$, of semi-computable numbers.
In fact, there exist left computable numbers $a,b\in\Sigma_{1}$ such that $x = a-b \notin \Csemi$.
This motivated \cite{WeihrauchZhengHierarchy_98} to define the following set of weakly computable numbers:

\begin{definition}[The set $\Cweak$ of weakly computable numbers]
\label{def:WeakComputableNr}
A number $x \in\RN$ is said to be weakly computable if there exist $a,b\in\Sigma_{1}$ such that $x = a-b$.
The set of all weakly computable numbers is denoted by $\Cweak$.
\end{definition}

\begin{remark}
$\Sigma_{1}$, $\Pi_{1}$, $\Csemi$ and $\Delta_{1}$ are proper subsets of $\Cweak$. Note that $\Cweak$ still belongs to the first layer of the Zheng--Weihrauch hierarchy.
\end{remark}

% ===== The second Layer of Zheng-Weihrauch =====
\subsection{The second layer in the Zheng--Weihrauch hierarchy}

A computable number $x$ was defined as having two particular properties, namely it has to be the limit of a computable sequence of rational numbers and the convergence has to be effective (cf. Def.~\ref{def:CompNumber}).
If one drops the second requirement, one obtains the fundamental class of non-computable functions in the second layer of the Zheng--Weihrauch hierarchy:

\begin{definition}[Recursively approximable numbers]
\label{def:RecursiveApprox}
A number $x \in \RN$ is said to be 
\emph{recursively approximable} if there exists a computable sequence $\left\{ x_{n} \right\}_{n\in\NN}$ of computable numbers that converges to $x$.
The set of all recursively approximable real numbers is denoted by $\Delta_{2}$. 
\end{definition}

It is clear from the definition that $\Delta_{1} \subset \Delta_{2}$, that $\Sigma_{1} \subset \Delta_{2}$, and that $\Pi_{1} \subset \Delta_{2}$.
Similarly as for $\Delta_{1}$, the set $\Delta_{2}$ can be written as the intersection of two sets.
To motivate the definition of these sets, let $\bsx = \left\{ x_{n} \right\}_{n\in\NN}$ be an arbitrary sequence of real numbers and define for every $N\in\NN$
\begin{equation*}
	\zeta_{N} = \inf \left\{ x_{k} : k\geq N \right\}
	\quad\text{and}\quad
	\xi_{N} = \sup \left\{ x_{k} : k\geq N \right\}\;.
\end{equation*}
Then $\left\{ \zeta_{N} \right\}_{N\in\NN}$ is a monotonically decreasing and $\left\{ \xi_{N} \right\}_{N\in\NN}$ a monotonically increasing sequence. 
It is well known that if the limit of $\left\{ \zeta_{N}\right\}$ and $\left\{ \xi_{N}\right\}$ are equal, i.e. if 
\begin{equation*}
	\lim_{N\to\infty} \zeta_{N} = \sup_{N\in\NN} \inf_{k\geq N} x_{k} = \liminf_{n\to\infty} x_{n}
	= \limsup_{n\to\infty} x_{n} = \inf_{N\in\NN} \sup_{k\geq N} x_{k} = \lim_{N\to\infty} \xi_{N} = x
\end{equation*}
then also $\left\{ x_{n} \right\}_{n\in\NN}$ converges (to the same limit $x$).
To make this result effective, Zheng and Weihrauch defined, similarly as in the first layer, the following two sets

\begin{definition}[$\Sigma_{2}$ and $\Delta_{2}$]
A number $x \in\RN$ belongs to $\Sigma_{2}$ if and only if there exists a computable double sequence $\left\{ \zeta_{n,k} \right\}_{n,k\in\NN} \subset\RNc$ such that
\begin{equation*}
	x = \textstyle\sup_{n\in\NN} \inf_{k\in\NN} \zeta_{n,k}\,.
\end{equation*}
A number $x \in\RN$ belongs to $\Pi_{2}$ if and only if there exists a computable double sequence $\left\{ \xi_{n,k} \right\}_{n,k\in\NN}\subset\RNc$ such that
\begin{equation*}
		x = \textstyle\inf_{n\in\NN} \sup_{k\in\NN} \xi_{n,k}\,.
\end{equation*}
\end{definition}

The following result \cite{ZhengWeihrauch_MathLog01} provides an effective version of the statement that $\left\{ x_{n} \right\}_{n\in\NN}$ converges if and only if $\lim\sup_{n\to\infty} x_{n} = \liminf_{n\to\infty} x_{n}$.

\begin{proposition}
\label{prop:Delta2isUnion}
A number $x \in\RN$ belongs to $\Delta_{2}$ if and only if $x \in \Sigma_{2} \cap \Pi_{2}$, that is $\Delta_{2} = \Sigma_{2} \cap \Pi_{2}$.
\end{proposition}
Proposition~\ref{prop:Delta2isUnion} is similar to Proposition~\ref{prop:Delta1isUnion} but for the second layer in the hierarchy.
As in the first layer, we have $\Delta_{2} \subsetneq \Sigma_{2}$ and $\Delta_{2} \subsetneq \Pi_{2}$.

% =================================================================
% ========== The Dirichlet problem on the unit Disk ===============
% =================================================================
\section{The Classical Dirichlet Problem on $\D$ -- Problem Statement}
\label{sec:ProblemStatment}

Let $f \in \C(\dD)$ be a real valued continuous function on $\dD$. 
The following classical problem, which goes back to \emph{George Green}, is of great importance in many areas of science and engineering:
Find $F : \DD\to \RN$ such that
\begin{equation}
\label{equ:DirProblem}
\begin{aligned}
	\big( \Delta F \big)(z) = 0\,, \quad&\text{for all}\ z\in\D\\
	\text{and}\qquad
	F(\zeta) = f(\zeta)\,,\quad&\text{for all}\ \zeta\in\dD\;.
\end{aligned}
\end{equation}
So one looks for a function $F$ that is harmonic in $\D$ and that coincides with the given function $f$ on the boundary $\dD$ of $\D$.
Historically, there are mainly two different ways to solve \eqref{equ:DirProblem}.

\subsection{Solution by Poisson integral}
\label{sec:DP_PoissionIntegral}

The first solution is due to \emph{Gustav Lejeune Dirichlet} who observed that the unique $F$ that satisfies the two conditions \eqref{equ:DirProblem} is given by the \emph{Poisson integral} of $f$, i.e.
\begin{equation}
\label{equ:SolutionPoissionIntegral}
	F(r\E^{\I\theta})
	= \left( \Op{P}_{r} f\right)(\E^{\I\theta}) 
	= \frac{1}{2\pi}\int^{\pi}_{-\pi} f(\E^{\I\tau})\, P_{r}(\theta - \tau)\,\d\tau\,,
	\quad 
\end{equation}
for $0\leq r < 1$ and $-\pi\leq \theta < \pi$,
with the \emph{Poisson kernel}
\begin{equation*}
	P_{r}(\theta) = \frac{1 - r^{2}}{1 - 2 r \cos(\theta) + r^{2}}\,,
\end{equation*}
and which satisfies $\lim_{r\to 1} F(r\E^{\I\theta}) = f(\E^{\I\theta})$ for all $\theta\in [-\pi,\pi)$.
If $f$ is written as a harmonic series as in \eqref{equ:FourierSeries} then its Poisson integral has the simple form
\begin{equation}
\label{equ:PoissionSeries}
	F(r\E^{\I\theta}) = \left( \Op{P}_{r} f\right)(\E^{\I\theta}) =
	\frac{a_{0}(f)}{2} + \sum^{\infty}_{n=1} r^{n} \left[ a_{n}(f)\cos(n\theta) + b_{n}(f)\sin(n\theta) \right].
\end{equation}
In honor of his contribution in the study of problem \eqref{equ:DirProblem}, it is now known as the \emph{Dirichlet problem}.

\subsection{Solution by Dirichlet's principle}
\label{sec:DP_varinaionalApproach}

\emph{Dirichlet} and \emph{Lord Kelvin} also proposed to solve \eqref{equ:DirProblem} by a variational method based on the minimization of \emph{Dirichlet's energy}:
Let $f\in\C(\dD)$ be a continuous function on $\dD$ and let
\begin{equation*}
	\M(f) = \left\{ u : \DD \to \RN\ :\ u\in\C(\DD) \cap \C^{1}(\D), u|_{\dD} = f \right\}
\end{equation*}
be the set of all real functions $u$ on $\DD$ that are continuous on $\DD$, continuously differentiable on $\D$, and that coincide with $f$ on $\dD$.
Now, \emph{Dirichlet's principle} states (see, e.g., \cites{Evans_PDEs,Mikhlin_VariationalMethods}) that the function $F$ that solves \eqref{equ:DirProblem} is obtained by finding the function in $\M(f)$ that has the smallest Dirichlet energy \eqref{equ:DirichletEnergy}.
So in order to find the solution of \eqref{equ:DirProblem} one may solve the following minimization problem
\begin{equation}
\label{equ:MinProblem}	
	F = \arg \inf_{u \in \M(f)} \Energy(u)\;.
\end{equation}
If we denote by 
\begin{equation*}
	\underline{\Energy}(f) = \inf_{u \in \M(f)} \Energy(u)
\end{equation*}
the minimum Dirichlet energy of functions in $\M(f)$, then it turns out that
\begin{equation}
\label{equ:MinEnergy}
	\underline{\Energy}(f) = \Energy(F) = \Op{E}(f)	
\end{equation}
i.e. the minimal Dirichlet energy $\underline{\Energy}(f)$ of functions in $\M(f)$ is equal to the Dirichlet energy $\Energy(f)$ of the boundary function $f$ and
there exists exactly one $F \in \M(f)$ that satisfies $\Energy(F) = \underline{\Energy}(f)$.
This minimizer is a solution of~\eqref{equ:DirProblem}.

% ---------- Problem statement ----------
\subsection{Problem statement}

Subsections~\ref{sec:DP_PoissionIntegral} and \ref{sec:DP_varinaionalApproach} described two approaches to solve \eqref{equ:DirProblem}.
From the analytical side, both ways give the same solution.
However, there might exist considerable differences as far as the computability of the algorithmic solutions is concerned.
In the following two sections we will show that the choice of which method is used to solve \eqref{equ:DirProblem} should depend on the set of boundary functions.
To this end, we notice first that \eqref{equ:MinEnergy} implies that the given boundary function $f\in\C(\dD)$ needs to have finite Dirichlet energy $\Energy(f) < \infty$, i.e. $f$ needs to be in $H^{1/2}(\dD)$.
Otherwise \eqref{equ:MinProblem} will not have a solution.
In addition, in order to pass the given boundary function $f\in\C(\dD) \cap H^{1/2}(\dD)$ to a numerical algorithm that solves \eqref{equ:MinProblem} or that computes the Poisson integral \eqref{equ:SolutionPoissionIntegral}, the function $f$ needs to be computable in some sense.
For the boundary function of the Dirichlet problem, it is natural to require either
\begin{equation*}
	f \in \Cc(\dD) \cap H^{1/2}(\dD)
	\quad\text{or}\quad
	f \in \C(\dD) \cap H^{1/2}_{\mathrm{c}}(\dD).
\end{equation*}
Thus one requires that $f$ can be effectively approximated by trigonometric polynomials either in the uniform norm or in the (Dirichlet) energy norm $\left\|\cdot\right\|_{H^{1/2}}$.
Both assumptions are not equivalent, i.e. there exist functions $f \in \Cc(\dD) \cap H^{1/2}(\dD)$ that do not belong to $\C(\dD) \cap H^{1/2}_{\mathrm{c}}(\dD)$ and vice versa.

In Section~\ref{sec:MainMinimizationProblem} we will show that there exist functions $f \in \Cc(\dD) \cap H^{1/2}(\dD)$ such that the minimum $\underline{\Energy}(f)$ of the minimization problem~\eqref{equ:MinProblem} is not a computable number.
For such a function, any numerical algorithm that attempts to solve \eqref{equ:MinProblem} will not be able to truncate the iteration effectively if a desired error bound is reached.
Thus solving the Dirichlet problem effectively by the optimization method \eqref{equ:MinProblem} is not possible for all functions in the class $\Cc(\dD) \cap H^{1/2}(\dD)$ of boundary functions.
In addition, Section~\ref{sec:MainMinimizationProblem} provides a complete characterization of the possible values $\underline{\Energy}(f)$ in terms of the Zheng--Weihrauch hierarchy.

Section~\ref{sec:MainDirichletIntegral} then studies the computation of the Poisson integral \eqref{equ:SolutionPoissionIntegral} for arbitrary functions from the set $\C(\dD) \cap H^{1/2}_{\mathrm{c}}(\dD)$.
It is shown that there exist functions $f\in \C(\dD) \cap H^{1/2}_{\mathrm{c}}(\dD)$ such that $f(\E^{\I\theta})$ is not computable at some computable position $\theta$. 
Thus solving the Dirichlet problem effectively based on the numerical integration of the Poisson integral is not possible for all function in the set $\C(\dD) \cap H^{1/2}_{\mathrm{c}}(\dD)$.
In addition, we derive lower and upper bounds on the degree of non-computability of the values $f(\E^{\I\theta})$ in the Zheng--Weihrauch hierarchy.

% =======================================================
% ========== Computation via Minimization ===============
% =======================================================
\section{Computation via the Minimization of Dirichlet energy}
\label{sec:MainMinimizationProblem}

For any given continuous boundary function $f\in\C(\dD)$ with finite Dirichlet energy $\Energy(f) < \infty$ the solution of the Dirichlet problem \eqref{equ:DirProblem} can be obtained by solving the minimization problem \eqref{equ:MinProblem}.
This section asks whether this minimization problem can always be effectively solved on a Turing machine.
We are going to show that this is not the case.
In fact there exist computable continuous boundary functions of finite Dirichlet energy such that the corresponding minimum energy $\underline{\Energy}(f)$ of \eqref{equ:DirProblem} is not a computable number.

\begin{theorem}
\label{thm:MinEnergInSigma1}
Let $f\in\Cc(\dD)$ with $\Energy(f) < \infty$ be arbitrary. Then 
\begin{equation}
\label{equ:Thm_inf1}
	\underline{\Energy}(f) = \inf_{u\in\M(f)} \Energy(u) \in \Pu\,.
\end{equation}
Conversely, to every $\alpha\in\Pu$, $\alpha > 0$, there exists an $f_{*} \in\Cc(\dD)$ with
\begin{equation}
\label{equ:Thm_inf2}
	\underline{\Energy}(f_{*}) = \inf_{u\in\M(f_{*})} \Energy(u) = \alpha\,.
\end{equation}
\end{theorem}

\begin{proof} %[Proof of Theorem~\ref{thm:MinEnergInSigma1}]
Let $f \in\Cc(\dD)$ with $\Energy(f)<\infty$ be arbitrary.
We write $f$ as a harmonic series \eqref{equ:FourierSeries}.
Because $f \in \Cc(\dD)$, the corresponding Fourier coefficients $\left\{ a_{n}(f) \right\}_{n\in\NN}$ and $\left\{ b_{n}(f) \right\}_{n\in\NN}$ form computable sequences of computable numbers \cite{PourEl_Computability}.
If we define
\begin{equation*}	
	E_{N} = \textstyle\frac{1}{2}\sum^{N}_{n=1} n \big[ \left|a_{n}(f)\right|^{2} + \left| b_{n}(f)\right|^{2} \big]\,,
	\quad N\in\NN\,,
\end{equation*}
then $\left\{ E_{N} \right\}_{N\in\NN}$ is a computable sequence of computable numbers that satisfies
\begin{align*}
	E_{N+1} \geq E_{N}\ \text{for all}\ N\in\NN
	\quad\text{and}\quad
	\lim_{N\to\infty} E_{N} = \Energy(f)\,.
\end{align*}
Thus, $\Energy(f) \in \Pu$ and therefore \eqref{equ:MinEnergy} implies \eqref{equ:Thm_inf1}.

To prove the second statement of the theorem, we choose $\alpha\in\Pu$ with $\alpha > 0$ arbitrary and construct an $f_{*}\in\Cc(\dD)$ with $\Energy(f_{*}) = \alpha$.
Then statement \eqref{equ:Thm_inf2} follows again from \eqref{equ:MinEnergy}.
To construct such an $f_{*}$, we define first the computable sequence $\left\{ k_{m} \right\}_{m\in\NN}$ by $k_{m} = m^{4}$. 
Then for every $m\geq 1$, we define the rational trigonometric polynomials
\begin{equation}
\label{equ:DefPhi_m}
	\varphi_{m}(\theta) = \cos(k_{m} \theta) \sum^{10}_{\ell = 1} \frac{\cos(\ell \theta)}{\ell}\,,
	\quad \theta\in [-\pi,\pi)\,.
\end{equation}
Using standard trigonometric identities, we write $\varphi_{m}$ as a harmonic series 
\begin{equation*}
	\varphi_{m}(\theta) = \sum^{k_{m}+10}_{n=k_{m}-10} a_{n}(\varphi_{m}) \cos(n \theta)\,,
	\quad \theta\in [-\pi,\pi)\,,
\end{equation*}
with the Fourier coefficients
\begin{equation*}
	a_{n}(\varphi_{m}) = 
	\left\{\begin{array}{ccl}
	\frac{1}{2\,(n-k_{m})} & \text{:} & n > k_{m}\\[1ex]
	0 & \text{:} & n=k_{m}\\[1ex]
	\frac{1}{2\,(k_{m}-n)} & \text{:} & n < k_{m}
	\end{array}\right.\;.
\end{equation*}
Before we proceed, we collect some properties of the functions $\varphi_{m}$.
First we notice that the Dirichlet energy of $\varphi_{m}$ is equal to the square of its $H^{1/2}$-norm and it is given by
\begin{align}
\label{equ:NormPhi_m}
	\left\|\varphi_{m}\right\|^{2}_{H^{1/2}}
	&= \Energy(\varphi_{m})
	= \textstyle\frac{1}{2}\sum^{k_{m}+10}_{n=k_{m}-10} n \left| a_{n}(\varphi_{m}) \right|^{2}
	= \textstyle\frac{1}{8}\sum^{k_{m}+10}_{n=k_{m}-10} \frac{n}{(n-k_{m})^{2}}\nonumber\\
	&= \textstyle\frac{1}{8} \left( \sum^{k_{m}-1}_{n=k_{m}-10} \frac{n}{(n-k_{m})^{2}} + \sum^{k_{m}+10}_{n=k_{m}+1} \frac{n}{(n-k_{m})^{2}} \right)\nonumber\\
	&= \textstyle\frac{1}{8} \left( \sum^{10}_{r=1} \frac{k_{m}-r}{r^{2}} + \sum^{10}_{r=1} \frac{k_{m} + r}{r^{2}} \right)
	= \frac{k_{m}}{4} \sum^{10}_{r=1} \frac{1}{r^{2}}
	= C_{0}\, m^{4}
\end{align}
with the positive constant $C_{0} = \textstyle\frac{1}{4} \sum^{10}_{r=1} \frac{1}{r^{2}}$.
Moreover, the definition \eqref{equ:DefPhi_m} of $\varphi_{m}$ shows that
\begin{equation}
\label{equ:UniformBoundPhim}
	\left|\varphi_{m}(\theta)\right|
	\leq \textstyle\sum^{10}_{\ell=1} \frac{1}{\ell} =: C_{\varphi}
\end{equation}
for all $m\in\NN$ and all $\theta\in [-\pi,\pi)$.
For every $m\geq 1$, let
\begin{equation*}
	\mathfrak{F}_{m} = \left\{ k_{m}-10, \dots, k_{m}-1, k_{m}+1, \dots, k_{m}+10 \right\}
\end{equation*}
be the set of indices of the non-vanishing Fourier coefficients of $\varphi_{m}$.
Since 
$k_{m+1} = (m+1)^{4} = k_{m} + 4 m^{3} + 6 m^{2} + 4 m +1$, we see that $k_{m+1} - k_{m} \geq 15$ for all $m\geq 1$.
This implies that $\mathfrak{F}_{m+1} \cap \mathfrak{F}_{m} = \emptyset$ for all $m\geq 1$.
Thus the functions $\varphi_{n}$ and $\varphi_{m}$ with $n\neq m$ always have disjoint sets of non-vanishing Fourier coefficients and therefore
\begin{equation}
\label{equ:NormAdd}
	\left\|\varphi_{n} + \varphi_{m}\right\|^{2}_{H^{1/2}} = \left\|\varphi_{n}\right\|^{2}_{H^{1/2}} + \left\|\varphi_{m}\right\|^{2}_{H^{1/2}}\,.	
\end{equation}

After these remarks, we proceed with the proof and notice that because $\alpha\in\Pu$, there exists a sequence $\left\{ \alpha_{n}\right\}_{n\in\NN}$ with
\begin{equation*}
	\alpha_{n+1} \geq \alpha_{n}\ \text{for all}\ n\in\NN
	\qquad\text{and}\qquad
	\lim_{n\to\infty} \alpha_{n} = \alpha\,.
\end{equation*}
Without loss of generality, we assume $\alpha_{1} = 0$ and define
\begin{equation*}
	d_{n} = \alpha_{n+1} - \alpha_{n}\,,
	\qquad n\geq 1\;.
\end{equation*}
Then we define for every $K\in\NN$ the function
\begin{equation*}
	f_{K}(\E^{\I\theta}) = \textstyle\sum^{K}_{m=1} \sqrt{d_{m}}\, \widetilde{\varphi}_{m}(\theta)\,,
	\qquad \theta\in [-\pi,\pi)\,,
\end{equation*}
with the normalized functions $\widetilde{\varphi}_{m} = \varphi_{m} / \left\|\varphi_{m}\right\|_{H^{1/2}}$, and the function
\begin{equation*}
	f_{*}(\E^{\I\theta}) = \textstyle\sum^{\infty}_{m=1} \sqrt{d_{m}}\, \widetilde{\varphi}_{m}(\theta)\,,
	\qquad \theta\in [-\pi,\pi)\,.
\end{equation*}
We are going to show that $f_{K} \to f_{*}$ as $K\to\infty$ effectively and uniformly on $\dD$.
Indeed, for every $\theta \in [-\pi,\pi)$, we have
\begin{align*}
	\left|f_{*}(\E^{\I\theta}) - f_{K}(\E^{\I\theta}) \right|
	&= \left| \textstyle\sum^{\infty}_{m=K+1} \sqrt{d_{m}}\, \widetilde{\varphi}_{m}(\theta) \right|
	\leq \sum^{\infty}_{m=K+1} \frac{\sqrt{d_m}}{\left\|\varphi_{m}\right\|_{H^{1/2}}} \left|\varphi_{m}(\theta)\right|\\
	&\leq C_{\varphi}\, \max_{m>K}\sqrt{d_{m}} \sum^{\infty}_{m=K+1} \frac{1}{ \left\|\varphi_{m}\right\|_{H^{1/2}} }
	\leq C_{1} \sum^{\infty}_{m=K+1} \frac{1}{m^{2}}	
\end{align*}
with the constant $C_{1} = \frac{C_{\varphi}}{\sqrt{C_{0}}}\, \left( \max_{m>K}\sqrt{d_{m}} \right)$ and using the properties \eqref{equ:NormPhi_m} and \eqref{equ:UniformBoundPhim} of the functions $\varphi_{m}$ to obtain the last line. 
With the usual estimate
\begin{equation*}
	\sum^{\infty}_{m=K+1} \frac{1}{m^{2}} 
	= \lim_{N\to\infty} \sum^{N}_{m=K+1} \frac{1}{m^{2}}
	\leq \lim_{N\to\infty} \int^{N}_{K}\frac{\d x}{x^{2}}
	= \frac{1}{K}
\end{equation*}
we finally obtain
\begin{equation*}
	\left|f_{*}(\E^{\I\theta}) - f_{K}(\E^{\I\theta}) \right| 
	\leq C_{1}\, \frac{1}{K}\,,
	\quad\text{for every}\ K\in\NN\,,
\end{equation*}
and uniformly for all $\theta\in [-\pi,\pi)$.
This shows that the computable sequence $\left\{ f_{K} \right\}_{K\in\NN}$ of computable trigonometric polynomials converges uniformly to $f_{*}$.
Consequently, by Definition~\ref{def:ComputableBanachSpace}, $f_{*} \in \Cc(\dD)$.

Finally, we determine the $H^{1/2}$-norm of $f_{*}$.
Using \eqref{equ:NormAdd} and the definition of the sequence $\left\{ d_{m} \right\}_{m\in\NN}$, we get
\begin{align*}
	\left\|f_{*}\right\|^{2}_{H^{1/2}}
	&= \left\| \textstyle\sum^{\infty}_{m=1} \sqrt{d_{m}}\, \widetilde{\varphi}_{m} \right\|^{2}_{H^{1/2}}
	=  \textstyle\sum^{\infty}_{m=1} d_{m}\, \left\| \widetilde{\varphi}_{m} \right\|^{2}_{H^{1/2}}
	= \textstyle\sum^{\infty}_{m=1} d_{m}\\[1ex]
	&= \textstyle\lim_{N\to\infty} \sum^{N}_{m=1} (\alpha_{m+1} - \alpha_{m})
	= \lim_{N\to\infty} \alpha_{N+1} - \alpha_{1} 
	= \alpha\;.
\end{align*}
Since $\Energy(f_{*}) = \left\|f_{*}\right\|^{2}_{H^{1/2}}$, this finishes the proof of the theorem.
\end{proof}

Theorem~\ref{thm:MinEnergInSigma1} provides a complete characterization of the degree of non-computability of the Dirichlet energy $\underline{\Energy}(f)$ of solutions of the Dirichlet problem for continuous computable boundary functions $f\in\Cc(\dD)$.
Its first part shows that for every computable continuous boundary function $f$ the minimal energy $\underline{\Energy}(f)$ always belongs to $\Pu$.
Thus $\underline{\Energy}(f)$ is always the limit of a monotonically increasing, computable sequence $\left\{ \xi_{n} \right\}_{n\in\NN}$ of computable numbers.
In particular, $\underline{\Energy}(f)$ is recursively approximable. 
Nevertheless, $\underline{\Energy}(f)$ is generally not a computable number. This follows from the second part of Theorem~\ref{thm:MinEnergInSigma1}.
Indeed, since $\Delta_{1} \subsetneq \Sigma_{1}$, we can choose $\alpha \in \Sigma_{1}\backslash \Delta_{1}$.
Then Part~2 of Theorem~\ref{thm:MinEnergInSigma1} shows that there exists a computable continuous function $f \in \Cc(\dD)$ of finite Dirichlet energy $\Energy(f)$ but such that $\Energy(f) = \underline{\Energy}(f)$ is not a computable number:

\begin{corollary}
\label{cor:NonComputableMinimum}
There exist functions $f \in \Cc(\dD)$ with $\Energy(f) < \infty$ such the minimum of the variational formulation \eqref{equ:MinProblem} of the Dirichlet problem has a non-computable solution $\underline{\Energy}(f)$.
\end{corollary}

\begin{remark}
\label{rem:NotH12computable}
If $f \in \Cc(\dD)$ is a function as in Corollary~\ref{cor:NonComputableMinimum} then its Dirichlet energy $\Energy(f) = \underline{\Energy}(f)$ is not a computable number.
Consequently, \eqref{equ:H12_norm} shows that $\left\|f\right\|_{H^{1/2}} \notin \Delta_{1}$.
Therefore, $f$ cannot be an $\H12$-computable function.
So for all boundary functions satisfying the condition of Corollary~\ref{cor:NonComputableMinimum}, the solution of the corresponding Dirichlet problem cannot be effectively computed on a Turing machine using the variational approach discussed on Section~\ref{sec:DP_varinaionalApproach}.
\end{remark}

If $f$ is a boundary function as in Corollary~\ref{cor:NonComputableMinimum},
then even although $\underline{\Energy}(f)$ is recursively approximable, it is only recursively approximable from below.
There does not exist any monotonically decreasing computable sequence that converges from above to $\underline{\Energy}(f)$.
In view of the variational formulation \eqref{equ:MinProblem} of the Dirichlet problem, this has important consequences, because most numerical algorithms that try to solve the convex minimization problem \eqref{equ:MinProblem} will produce a sequence $\left\{ x_{n} \right\}_{n\in\NN}$ that converges to $\underline{\Energy}(f)$ from above.
Then Theorem~\ref{thm:MinEnergInSigma1} implies that any such approximation sequence cannot be computable.

\begin{corollary}
\label{cor:NoSequenceOfUpperBounds}
Let $f \in \Cc(\dD)$ with $\Energy(f) < \infty$ as in Corollary~\ref{cor:NonComputableMinimum} and let $\bsx = \left\{ x_{n} \right\}_{n\in\NN}$ be a monotonically decreasing sequence with $\lim_{n\to\infty} x_{n} = \underline{\Energy}(f)$.
Then $\bsx$ is not a computable sequence of computable numbers.
\end{corollary}
So in order to construct any numerical procedure to solve the variational problem \eqref{equ:Thm_inf1}, one would like to construct a sequence $\left\{ F_{n} \right\}_{n\in\NN} \in \M(f)$ with
$\Energy(F_{n+1}) \leq \Energy(F_{n})$, for all $n\in\NN$ and with
\begin{align}
\label{equ:convergence}
	\lim_{n\to\infty} \Energy(F_{n}) &= \underline{\Energy}(f)\,.
\end{align} 
In fact, it would be sufficient that $\left\{ F_{n} \right\}_{n\in\NN}$ satisfies \eqref{equ:convergence}. The monotonicity of the sequence is not necessarily needed since from any sequence $\left\{ \Energy(F_{n}) \right\}_{n\in\NN}$ that satisfies \eqref{equ:convergence} we can always construct a monotonically decreasing subsequence $\left\{ x_{N}\right\}_{N\in\NN}$ by
\begin{equation}
\label{equ:Construction_xN}
	x_{N} = \min \big\{ \Energy(F_{n}) : n=1,2,\dots,N \big\}\,,
	\quad N\in\NN\,,
\end{equation}
that converges to $\underline{\Energy}(f)$ as $N\to\infty$.
The next theorem shows that there exist computable continuous functions $f_{*} \in \Cc(\dD)$ such that the set $\M(f_{*})$ does not contain any useful computable structure for the computation of $\underline{\Energy}(f_{*})$ in the sense of effective analysis. 

\begin{theorem}
Let $\alpha \in \Pu\backslash\Delta_{1}$, $\alpha >0$ be arbitrary and let $f_{*}\in\Cc(\dD)$ be the corresponding function as constructed in Theorem~\ref{thm:MinEnergInSigma1}.
Then there exists no computable sequence $\left\{ F_{n} \right\}_{n\in\NN} \subset \mathcal{M}(f_{*})$ of computable continuous functions such that $\Energy(F_{n}) \in \RNc$ for every $n\in\NN$ and such that
\begin{equation*}
	\lim_{n\to\infty} \Energy(F_{n}) = \underline{\Energy}(f_{*})\;.
\end{equation*}
\end{theorem}

\begin{proof}
Let $\alpha \in \Pu$, $\alpha >0$ but $\alpha\notin\Delta_{1}$ be arbitrary.
Then there exists a computable sequence $\left\{ \alpha_{n} \right\}_{n\in\NN}$ of computable numbers that is monotonically increasing and converges to $\alpha$.
Assume there exists a computable sequence $\left\{ F_{n} \right\}_{n\in\NN} \subset \mathcal{M}(f_{*})$ of computable continuous functions that satisfies \eqref{equ:convergence}.
Then we have $\Energy(F_{n}) \geq \underline{\Energy}(f_{*})$ for all $n\in \NN$ and so \eqref{equ:Construction_xN} defines a monotonically decreasing, computable sequence of computable numbers that converges to $\underline{\Energy}(f_{*}) = \alpha$. This implies that $\alpha\in\Pi_{1}$ and consequently $\alpha \in  \Sigma_{1} \cap \Pi_{1} = \Delta_{1}$. This contradicts the assumption $\alpha\notin\Delta_{1}$ and thus proves the theorem.
\end{proof}

Finally, we would like to discuss the statement of Theorem~\ref{thm:MinEnergInSigma1} in some more detail.
According to the second part of the theorem, if $\alpha \in \Sigma_{1}\backslash\RNc$ then the Dirichlet energy $\underline{\Energy}(f_{*})$ of the corresponding function $f_{*}$ is not a computable number.
However, $f_{*}$ is a computable continuous functions.
So according to the definition of the computable structure in $\Cc(\dD)$ (cf. Definition~\ref{def:ComputableBanachSpace}), there exists a computable sequence $\left\{ \varphi^{*}_{n} \right\}_{n\in\NN}$ of trigonometric polynomials that effectively converges to $f_{*}$ in the maximum norm, i.e. such that
\begin{equation*}
	\left\| f_{*} - \varphi^{*}_{n}\right\|_{\infty} < 2^{-n}\,,
	\qquad\text{for all}\ n\in\NN\,.
\end{equation*}
The trigonometric polynomials $\varphi^{*}_{n}$ have the form \eqref{equ:TrigPoly}, with the sequences of Fourier coefficients $\left\{ a_{n}(k) \right\}_{k\in\NN}$ and $\left\{ b_{n}(k) \right\}_{k\in\NN}$ that are computable sequences of computable numbers.
Therefore it is clear that $\left\{ \underline{E}(\varphi^{*}_{n}) \right\}_{n\in\NN}$, with
\begin{equation*}
	\underline{E}(\varphi^{*}_{n}) = \frac{1}{2}\sum^{d(n)}_{k=1} n\left( \left|a_{n}(k)\right|^{2} + \left|n_{n}(k)\right|^{2} \right)
\end{equation*}
 is a computable sequence of computable numbers.
Moreover, by the continuity of the energy functional (i.e. by the continuity of the $\H12$-norm), we have
\begin{equation}
\label{equ:ConvergenceInEnergy}
	\lim_{n\to\infty} \underline{\Energy}(f_{*} - \varphi^{*}_{n})
	= \underline{\Energy}(f_{*}) - \lim_{n\to\infty} \underline{\Energy}(\varphi^{*}_{n}) = 0
\end{equation}
which shows that $\left\{ \underline{E}(\varphi^{*}_{n}) \right\}_{n\in\NN}$ converges to $\underline{\Energy}(f_{*})$.
However, since according to Theorem~\ref{thm:MinEnergInSigma1}, $\underline{\Energy}(f_{*})$ is not a computable number, this convergence cannot be effective.
Thus the computable sequence of computable numbers $\left\{ \underline{E}(\varphi^{*}_{n}) \right\}_{n\in\NN}$ converges to $\underline{\Energy}(f_{*})$ but it is impossible to effectively control the approximation error of this convergence.
 
The previous arguments show that $\left\{ \varphi^{*}_{n} \right\}_{n\in\NN}$ is also a computable sequence in $\H12(\dD)$, and relation \eqref{equ:ConvergenceInEnergy} can equivalently be written as $\lim_{n\to\infty} \left\|f_{*} - \varphi^{*}_{n}\right\|_{\H12} = 0$.
Thus $\left\{ \varphi^{*}_{n} \right\}_{n\in\NN}$ converges to $f_{*}$ in the norm of $\H12(\dD)$. 
However, since this convergence is not effective in $\H12_{\mathrm{c}}(\dD)$ the function $f_{*}$ cannot be a $\H12$-computable function (cf. also Remark~\ref{rem:NotH12computable}).

% ==========================================================
% ========== Computation via Dirichlet Integral ============
% ==========================================================
\section{Computation via Integrating the Poisson Integral}
\label{sec:MainDirichletIntegral}

The Dirichlet problem \eqref{equ:DirProblem} has a unique solution if the boundary function $f$ is continuous and has finite Dirichlet energy.
Assume that $f$ can additionally be effectively approximated by trigonometric polynomials in the $H^{1/2}$-norm,
i.e. let $f \in \H12_{\mathrm{c}}(\dD) \cap \C(\dD)$ be a continuous functions on $\dD$ with finite Dirichlet energy.
Is it then true that $f$ is Banach--Mazur computable, i.e. is it true that $f(\E^{\I\theta}) \in \RNc$ for all computable numbers $\theta\in\RNc$? 
The negative answer follows immediately from the following theorem.

\begin{theorem}
\label{thm:MainThm2}
Let $\theta\in\RNc$ be arbitrary.\\
1) To every $\alpha \in \Cweak$ there exists an $f_{*} \in \H12_{\mathrm{c}}(\dD) \cap \C(\dD)$ with $f_{*}(\E^{\I\theta}) = \alpha$.\\
2) Every $f \in L^{2}_{\mathrm{c}}(\dD) \cap \C(\dD)$ satisfies $f(\E^{\I\theta}) \in \Delta_{2}$.
\end{theorem}

The first part of Theorem~\ref{thm:MainThm2} shows that every weakly computable number $\alpha\in\Cweak$ (cf. Def.~\ref{def:WeakComputableNr}) is the value of a function $f_{*} \in \H12_{\mathrm{c}}(\dD) \cap \C(\dD)$.
Since $\Delta_{1} \subsetneq \Cweak$, it is always possible to choose an $\alpha\in \Cweak$ that is not computable.
Then for the corresponding function $f_{*}$ the value $f_{*}(\E^{\I\theta})$ is not a computable number, i.e. $f_{*}$ does not map all computable numbers onto computable numbers.
Therefore $f_{*}$ is not Banach--Mazur computable and consequently it is also not a computable continuous function.
This answers the question posed at the beginning of this section:

\begin{corollary}
\label{cor:NotBanachMazurComputable}
To every $\theta\in\RNc$ there exists an $f_{*}\in \H12_{\mathrm{c}}(\dD) \cap \C(\dD)$ that is not Banach-Mazur computable with $f_{*}(\E^{\I\theta}) \notin \RNc$.
\end{corollary}
Part 1) of Theorem~\ref{thm:MainThm2} provides a lower bound on the non-computability (measured in the Zheng--Weihrauch hierarchy) of the function values in the set $\H12_{\mathrm{c}}(\dD) \cap \C(\dD)$.
It shows that $\H12_{\mathrm{c}}(\dD) \cap \C(\dD)$ contains "bad" functions $f_{*}$ such that $f_{*}(\E^{\I\theta})$ is only weakly computable at some computable points $\theta\in\dD$. 
Conversely, Part 2) of Theorem~\ref{thm:MainThm2} provides an upper bound on the non-computability of functions in $\H12_{\mathrm{c}}(\dD) \cap \C(\dD)$.
It shows that $f(\E^{\I\theta})$ always belongs to $\Delta_{2}$, i.e. 
$f(\E^{\I\theta})$ is always at least recursively approximable but it will never be less computable than $\Delta_{2}$. Thus $f(\E^{\I\theta})$ will never belong to $\Delta_{n}$ with $n>2$.
In fact, Theorem~\ref{thm:MainThm2} shows that this is even true for all functions in the larger set $L^{2}_{\mathrm{c}}(\dD) \cap \C(\dD) \supsetneq H^{1/2}_{\mathrm{c}}(\dD) \cap \C(\dD)$.

According to Corollary~\ref{cor:NotBanachMazurComputable} %
there always exist functions $f_{*} \in H^{1/2}_{\mathrm{c}}(\dD) \cap \C(\dD)$ such that the value $f_{*}(\E^{\I\theta})$ is not computable.
Thus $f(\E^{\I\theta}) \notin \Sigma_{1}\cap\Pi_{1}$ which means that the value $f_{*}(\E^{\I\theta})$ cannot be approximated from below \emph{and} above.
If $f_{*}(\E^{\I\theta})$ would belong either to $\Sigma_{1}$ or to $\Pi_{1}$ then it would at least be possible to approximate $f_{*}(\E^{\I\theta})$ from below \emph{or} above by a monotonically increasing (decreasing) computable sequence.
More precisely, if $f_{*}(\E^{\I\theta})$ would be in $\Sigma_{1}$ then there would exist a computable sequence $\left\{ \zeta_{n} \right\}_{n\in\NN}$ with
\begin{equation}
\label{equ:seqenceZeta}
	\zeta_{n+1} \geq \zeta_{n}\,,\quad \text{for all}\ n\in\NN 
	\qquad\text{and}\qquad  
	\lim_{n\to\infty} \zeta_{n} = f(\E^{\I\theta})\,,
\end{equation}
and if $f_{*}(\E^{\I\theta})$ would be in $\Pi_{1}$ then there would exist a computable sequence $\left\{ \xi_{n} \right\}_{n\in\NN}$ with
\begin{equation}
\label{equ:seqenceXi}
	\xi_{n+1} \leq \xi_{n}\,,\quad \text{for all}\ n\in\NN 
	\qquad\text{and}\qquad  
	\lim_{n\to\infty} \xi_{n} = f(\E^{\I\theta})\,.
\end{equation}
Now, the important point to notice is that the statement of Part 1) of Theorem~\ref{thm:MainThm2} is even stronger. Namely it shows that $f_{*}(\E^{\I\theta})$ can be even less computable.
Since $f_{*}(\E^{\I\theta})$ is generally only in $\Cweak$ and because $\Sigma_{1} \cup \Pi_{1}\subsetneq \Cweak$, we can choose the $\alpha \in \Cweak$ in the first part of Theorem~\ref{thm:MainThm2} such that $f_{*}(\E^{\I\theta}) = \alpha \notin \Sigma_{1} \cup \Pi_{1}$.
Then $f_{*}(\E^{\I\theta})$ can be approximated neither from below \emph{nor} from above.
So in this situation there exists neither a computable sequence $\left\{ \zeta_{n} \right\}_{n\in\NN}$ that satisfies \eqref{equ:seqenceZeta} 
nor a computable sequence $\left\{ \xi_{n} \right\}_{n\in\NN}$ that satisfies \eqref{equ:seqenceXi}.
In other words, there exists no computable sequence of lower bounds for $f(\E^{\I\theta})$ \emph{and} no computable sequence of upper bounds for $f(\E^{\I\theta})$ that are asymptotically sharp.
Similarly as in Section~\ref{sec:MainMinimizationProblem}, we have the following stronger corollary of Theorem~\ref{thm:MainThm2}.

\begin{corollary}
Let $\alpha \in \Cweak$, $\alpha \notin \Sigma_{1} \cup \Pi_{1}$ be arbitrary and let $f_{*} \in \H12_{\mathrm{c}}(\dD) \cap \C(\dD)$ be the corresponding function as in Part~1) of Theorem~\ref{thm:MainThm2}.
\begin{enumerate}
\item If $\left\{ \zeta_{n} \right\}_{n\in\NN}$ is a sequence that satisfies \eqref{equ:seqenceZeta} then $\left\{ \zeta_{n} \right\}_{n\in\NN}$ is not a computable sequence of computable numbers.
\item If $\left\{ \xi_{n} \right\}_{n\in\NN}$ is a sequence that satisfies \eqref{equ:seqenceXi} then $\left\{ \xi_{n} \right\}_{n\in\NN}$ is not a computable sequence of computable numbers.
\end{enumerate}
\end{corollary}

Compare this statement with Corollary~\ref{cor:NoSequenceOfUpperBounds}.
It shows that the situation for computing $f(\E^{\I \theta})$ for functions in $\H12_{\mathrm{c}}(\dD) \cap \C(\dD)$ is worse compared with computing the Dirichlet energy of functions in $\Cc(\dD)$. 
In the later case, the Dirichlet energy can at least always be approximated by a sequence of lower bounds.
For the value of $f(\E^{\I \theta})$ there exist in general neither a sequence of upper nor of lower bounds.

\begin{proof}[Theorem~\ref{thm:MainThm2}]
Part 1)
Without loss of generality, we prove the statement only for $\theta = 0$. Adapting the proof to arbitrary $\theta\in\RNc$ is obvious.
For every $M\in\NN$, $M\geq 2$, we consider the function
\begin{equation*}
	\varphi_{M}(\theta)
	= \textstyle\sum^{M}_{n=2} \frac{1}{n\log n} \cos(n \theta)\,,
	\quad \theta\in [-\pi,\pi)\,,
\end{equation*}
and define the positive constant 
\begin{equation*}
	C(M) = \varphi_{M}(0) = \textstyle\sum^{M}_{n=2} \frac{1}{n\log n}\,,
	\quad M\geq 2\;.
\end{equation*}
which is lower bounded by
\begin{align*}
	C(M) &= \textstyle\sum^{M}_{n=2} \frac{1}{n\log n}
	\geq \int^{M+1}_{2} \frac{\d x}{x \log x}
	= \int^{\log(M+1)}_{\log(2)}\frac{\d\tau}{\tau}\\
	&= \log\log (M+1) - \log\log 2
	= \log\tfrac{\log(M+1)}{\log 2}\,.
\end{align*}
Moreover, the definition of $\varphi_{M}$ immediately shows that
\begin{equation}
\label{equ:UpperBoundInftyNorm}
	\left|\varphi_{M}(\theta)\right|
	\leq \textstyle\sum^{M}_{n=2} \frac{\left|\cos(n\theta)\right|}{n\log n} 
	\leq C(M)\,,
	\ \text{for all}\ \theta\in [-\pi,\pi)\;,
\end{equation}
which implies in particular that $\left\|\varphi_{M}\right\|_{\infty} \leq C(M)$.
For the $H^{1/2}$-norm of $\varphi_{M}$, we get the following upper bound:
\begin{align*}
	\left\|\varphi_{M}\right\|^{2}_{H^{1/2}}
	&= \textstyle\sum^{M}_{n=2} n \left| a_{n}(\varphi_{M})\right|^{2}
	= \sum^{M}_{n=2} \frac{1}{n (\log n)^{2}}
	= \textstyle\frac{1}{2(\log 2)^{2}} + \sum^{M}_{n=3} \frac{1}{n (\log n)^{2}}\\
	&\leq \textstyle\frac{1}{2(\log 2)^{2}} + \int^{M}_{2} \frac{\d x}{x (\log x)^{2}}
	\leq \textstyle\frac{1}{2(\log 2)^{2}} + \int^{\infty}_{\log 2} \frac{\d\tau}{\tau^{2}}
	= \textstyle\frac{1}{2(\log 2)^{2}} + \frac{1}{\log 2}
\end{align*}
and we denote by $K_{3} \in \NN$ the smallest natural that satisfies
\begin{equation*}
	\textstyle\frac{1}{\log 2}\left(1 + \frac{1}{2\log 2}\right)
	= \textstyle\frac{1}{2(\log 2)^{2}} + \frac{1}{\log 2} \leq K_{3}\;.
\end{equation*}
Finally, we define
\begin{equation*}
	M(k) = 2^{2^{k^{2}}}\,,
	\quad k\in\NN\;.
\end{equation*}
Then $\left\{ M(k)\right\}_{k\in\NN}$ is a computable sequence of computable numbers.

After these preparations, let $\alpha\in \Cweak$ be an arbitrary weakly computable number.
Then there exists a computable sequence $\left\{ \alpha_{n} \right\}_{n\in\NN}$ of computable numbers (cf. \cite{WeihrauchZhengHierarchy_98}*{Theorem~3})
with
\begin{equation}
\label{equ:properties_alphaSequence}
	\textstyle\sum^{\infty}_{n=1} \left|\alpha_{n+1} - \alpha_{n} \right| < +\infty
	\quad\text{and}\quad
	\lim_{n\to\infty} \alpha_{n} = \alpha\;.
\end{equation}
Without loss of generality, we assume $\alpha_{1} = 0$ and define for all $n\geq 1$ the numbers $d_{n} = \alpha_{n+1} - \alpha_{n}$, and $K_{4} \in \NN$ will denote the smallest natural number that satisfies
\begin{equation*}
	\max_{n\in\NN} \left|d_{n}\right| \leq K_{4}\;.
\end{equation*}
Therewith, we define the function
\begin{equation*}
	f_{*}(\E^{\I\theta}) = \sum^{\infty}_{n=1} d_{n} \frac{\varphi_{M(n)}(\theta)}{C( M(n) )}\,,
	\qquad \theta\in [-\pi,\pi)\;,
\end{equation*}
and for every $K\in\NN$, the partial sum
\begin{equation*}
	f_{K}(\E^{\I\theta}) = \sum^{K}_{n=1} d_{n} \frac{\varphi_{M(n)}(\theta)}{C( M(n) )}\,,
	\qquad \theta\in [-\pi,\pi)\;.
\end{equation*}
This defines a computable sequence $\left\{ f_{K} \right\}_{K\in\NN}$ of trigonometric polynomials.

First we show that $f_{*} \in \C(\dD)$.
Applying the triangle inequality, \eqref{equ:UpperBoundInftyNorm} and the definition of $d_{n}$, we get
\begin{align*}
	\left|f_{*}(\E^{\I\theta}) - f_{K}(\E^{\I\theta}) \right|
	&= \left| \textstyle\sum^{\infty}_{n=K+1} d_{n} \frac{\varphi_{M(n)}(\theta)}{C( M(n) )} \right|\\[1ex]
	&\leq \textstyle\sum^{\infty}_{n=K+1} \left| d_{n} \right|
	= \sum^{\infty}_{n=K+1} \left| \alpha_{n+1} - \alpha_{n} \right|\,,
\end{align*}
for every $\theta \in [-\pi,\pi)$.
Then \eqref{equ:properties_alphaSequence} implies that the right hand side converges to zero as $K\to\infty$.
This shows that 
$$\textstyle\lim_{K\to\infty} \left\|f_{*} - f_{K} \right\|_{\infty} = 0\,,$$ 
proving that $f$ is a continuous function.

With respect to the convergence in the $H^{1/2}$-norm, we get
\begin{multline*}
	\left\| f_{*} - f_{K} \right\|_{H^{1/2}}
	= \left\| \sum^{\infty}_{n=K+1} d_{n} \frac{\varphi_{M(n)}}{C( M(n) )} \right\|_{H^{1/2}}\\
	\leq \sum^{\infty}_{n=K+1} \left\| d_{n} \right\| \tfrac{1}{C( M(n) )} \left\| \varphi_{M(n)} \right\|_{H^{1/2}}
	< K_{3} K_{4} \sum^{\infty}_{n=K+1} \frac{1}{\log\tfrac{\log M(n)}{\log 2}}\\
	= K_{3} K_{4} \sum^{\infty}_{n=K+1} \frac{1}{\log (2^{n^2})}
	= \frac{K_{3} K_{4}}{\log 2} \sum^{\infty}_{n=K+1} \frac{1}{n^2}
	= \frac{K_{3} K_{4}}{\log 2} \lim_{N\to\infty}\sum^{N}_{n=K+1} \frac{1}{n^2}\\
	\leq \frac{K_{3} K_{4}}{\log 2} \lim_{N\to\infty} \int^{N}_{K} \frac{\d x}{x^2}
	=  \frac{K_{3} K_{4}}{\log 2} \frac{1}{K}\;.
\end{multline*}
So if $K_{5} \in\NN$ denotes the smallest natural number that satisfies $\frac{K_{3} K_{4}}{\log 2} \leq K_{5}$, we thus have
\begin{equation*}
	\left\|f_{*} - f_{K} \right\|_{H^{1/2}} 
	\leq K_{5}\, \frac{1}{K}\,,
	\quad\text{for all}\ K\in\NN\,.
\end{equation*}
This shows that $\left\{ f_{K} \right\}_{K\in\NN}$ effectively converges to $f_{*}$ in the $H^{1/2}$-norm.
Since $\left\{ f_{K} \right\}_{K\in\NN}$ is a computable sequence of trigonometric polynomials in $H^{1/2}$, we thus have shown that $f_{*} \in H^{1/2}_{\mathrm{c}} \cap \C(\dD)$.

Part 2)
Let $f\in L^2_{\mathrm{c}}(\dD)\cap\C(\dD)$ and $\theta \in [-\pi,\pi)$ be arbitrary.
We have to show that there exists a computable sequence of computable numbers $\left\{x_{n}\right\}_{n\in\NN} \subset \RN$ that converges to $f(\E^{\I\theta})$.
Since $f \in \C(\dD)$ its Poisson integral can be written as in \eqref{equ:PoissionSeries} and where the sum converges uniformly for all $0 \leq r < 1$ and $\theta\in [-\pi,\pi)$.
Moreover, since $f \in L^2_{\mathrm{c}}(\dD)$, its Fourier coefficients form two computable sequences $\left\{ a_{n}(f) \right\}_{n\in\NN}$ and $\left\{ b_{n}(f) \right\}_{n\in\NN}$ of computable numbers in $\ell^{2}$.
We denote by $K_{1} \in \NN$ the smallest non-negative integer that satisfies
\begin{equation*}
	\max\Big( \max_{n\in\NN} \left|a_{n}(f)\right| , \max_{n\in\NN} \left|b_{n}(f)\right| \Big )
	\leq K_{1}\;.
\end{equation*}
For any $M\in\NN$, we write 
\begin{equation*}
	\left( \Op{P}^{M}_{r} f \right)(\E^{\I\theta}) =
	\frac{a_{0}(f)}{2} + \sum^{M}_{n=1} r^{n} \left[ a_{n}(f)\cos(n\theta) + b_{n}(f)\sin(n\theta) \right].
\end{equation*}
for the partial sum of the Poisson integral in \eqref{equ:PoissionSeries}.
Then for every $M\in\NN$, we have
\begin{multline*}
	\left| \left( \Op{P}_{r} f \right)(\E^{\I\theta}) - \left( \Op{P}^{M}_{r} f \right)(\E^{\I\theta}) \right|\\
	\leq \sum^{\infty}_{n=M+1}\!\!\! r^{n} \left(\left|a_{n}(f)\right| + \left|b_{n}(f)\right|\right)
	\leq K_{1}\!\!\! \sum^{\infty}_{n=M+1}\!\!\! r^{n}
	= K_{1} \frac{r^{M+1}}{1-r}\;.
\end{multline*} 
Next, we define $r_{k} = 1-1/k$ and $M_{k} = k^2- k$ for $k\in\NN$.
Then the previous inequality gives
\begin{equation*}
	\left| \left( \Op{P}_{r_k} f \right)(\E^{\I\theta}) - \left( \Op{P}^{M_k}_{r_k} f \right)(\E^{\I\theta}) \right|
	< K_{1}\, k \left(1 - \tfrac{1}{k}\right)^{k^2+k}\,.
\end{equation*}
To further upper bound the right hand side, we use the known inequalities
$(1-\frac{1}{k})^{k+1} < \E^{-1} < (1-\frac{1}{k})^{k}$ and $\E^{-1} < \frac{1}{2}$. 
This gives 
\begin{equation*}
	\left( 1-\tfrac{1}{k} \right)^{k(k+1)}
	< \E^{-k} < 2^{-k}\,,
\end{equation*}
and therewith, we finally obtain
\begin{equation}
\label{equ:PoisEffConv}
	\left| \left( \Op{P}_{r_k} f \right)(\E^{\I\theta}) - \left( \Op{P}^{M_k}_{r_k} f \right)(\E^{\I\theta}) \right|
	< K_{1}\, \frac{k}{2^{k}}\,.
\end{equation}
Now, we define
\begin{equation*}
	x_{k} = \left( \Op{P}^{M_k}_{r_k} f \right)(\E^{\I\theta})\,,
	\qquad k\in\NN\,.
\end{equation*}
Then its is clear that $\left\{ x_{k} \right\}_{k\in\NN}$ is a computable sequence of computable numbers, for which we have
\begin{align*}
	\left| f(\E^{\I\theta}) - x_{k}\right|
	&= \left|f(\E^{\I\theta}) - \left( \Op{P}^{M_k}_{r_k} f \right)(\E^{\I\theta}) \right|\\
	&\leq \left|f(\E^{\I\theta}) - \left( \Op{P}_{r_k} f \right)(\E^{\I\theta}) \right| + \left| \left( \Op{P}_{r_k} f \right)(\E^{\I\theta}) - \left( \Op{P}^{M_k}_{r_k} f \right)(\E^{\I\theta}) \right|\;.
\end{align*}
The first term on the right hand side converges to zero as $k\to\infty$ because the Poisson integral $\Op{P}_{r} f$ of a continuous function $f$ converges uniformly to $f$ as $r\to 1$, %(cf., e.g., \cites{Hoffman,Rudin,Zygmund,PB_Book_AdvTopics}), 
and the second term (effectively) converges to zero because of \eqref{equ:PoisEffConv}.
Thus, the computable sequence $\left\{ x_{k} \right\}_{k\in\NN}$ of computable numbers converges (not necessarily effective) to $f(\E^{\I\theta})$.
This proves that $f(\E^{\I\theta}) \in \Delta_{2}$.
\end{proof}

% ==================================
% ========== Conclusion ============
% ==================================
\section{Conclusions and Discussions}
\label{sec:Conclusion}

The classical Dirichlet problem \eqref{equ:DirProblem} on the unit disk can be solved using two different major approaches.
The first approach, based on Dirichlet principle, searches for the function of smallest Dirichlet energy that satisfies the boundary condition.
The second approach, based the maximum principle of harmonic functions, computes the Poisson integral of the boundary function.
We have shown that in both approaches a computable boundary function yields generally a non-computable solution.
Using the Zheng--Weihrauch hierarchy, we studied the following two problems.

\begin{enumerate}
\item
Consider the Dirichlet problem \eqref{equ:DirProblem} with a computable continuous boundary function $f$.
What are possible values for the degree of non-computability of the corresponding Dirichlet energy $\underline{\Energy}(f)$?
\item
Consider the Dirichlet problem \eqref{equ:DirProblem} with a boundary function $f$ that is continuous on $\dD$ and that is a $\H12(\dD)$-computable function.
Let $\zeta\in\dD$ be an arbitrary computable point on the unit circle.
What are possible values for the degree of non-computability of the function value $f(\zeta)$ at $\zeta$?
\end{enumerate}
To tackle these two problems, we needed to solve the following two subproblem for each of the above problems:
\begin{enumerate}
\item[a)]
Find an \emph{upper bound} on the degree of non-computability, i.e. find the layer in the Zheng--Weihrauch hierarchy such that all possible solutions belong to this layer.
\item[b)] 
Find a \emph{lower bound} on the degree of non-computability.
To this end, one chooses a particular layer in the Zheng--Weihrauch hierarchy.
Then one needs to find for a number $\alpha$ in this layer (that does not belong to a lower layer of the hierarchy) an input such that the result of our operation on this input is $\alpha$.
The lowest layer for which such a number $\alpha$ exists is the desired lower bound.
\end{enumerate}
For Problem~1, we were able to completely solve both subproblems and we have shown that the lower and upper bound coincide and that it is equal to $\Sigma_{1}$.
Thus the degree of non-computability of the Dirichlet energy $\underline{\Energy}(f)$ for computable continuous functions $f$ lies in the first layer of the Zheng--Weihrauch hierarchy.
For Problem~2, we have shown that the lower bound on the degree of non-computability is $\Cweak$, i.e. we showed that for every computable $\zeta\in\dD$ and every $\alpha\in\Cweak$ there exists a function $f_{*} \in \H12_{\mathrm{c}}(\dD) \cap \C(\dD)$ such that $f(\zeta) = \alpha$. As an upper bound on the degree of non-computability, we obtained $\Delta_{2}$, i.e. $f(\zeta)$ always belongs to $\Delta_{2}$ for every $f \in \H12_{\mathrm{c}}(\dD) \cap \C(\dD)$. 
Since $\Cweak \subsetneq \Delta_{2}$, it is not quite clear weather $\Delta_{2}$ is indeed the correct degree of non-computability for Problem 2, or whether our upper bound is not good enough.

% =============================================
% ========== Appendix =========================
% =============================================
\appendix
\section{The complete Zheng--Weihrauch hierarchy}
\label{sec:AppZW}

Continuing the procedure described in Section~\ref{sec:ZWHierarchy} and alternately applying $\sup$ and $\inf$ to computable sequence of multiple indices, Zheng and Weihrauch defined in \cite{ZhengWeihrauch_MathLog01} the following hierarchy of real numbers.

\begin{definition}
For $n\in\NN$ let $\bm = (m_{1}, \dots, m_{n}) \in \NN^{n}$. 
Then for every $n\in\NN$, the sets $\Sigma_{n} \subsetneq \RN$, $\Pi_{n} \subsetneq \RN$, and $\Delta_{n} \subsetneq \RN$ are defined as follows:
\begin{itemize}
\item A number $x\in\RN$ belongs to $\Sigma_{n}$ if and only if there exists a $n$-fold computable sequence $\left\{ r_{\bm} \right\}_{\bm\in\NN^{n}} \subset \QN$ of rational numbers such that
\begin{equation*}
	x = \sup_{m_{1} \in \NN} \inf_{m_{2} \in \NN} \sup_{m_{3} \in \NN} \dots \underset{\bm \in \NN}\Theta r_{\bm}
\end{equation*}
and where $\Theta$ stands for the infimum if $n$ is even and for the supremum if $n$ is odd.
\item A number $x\in\RN$ belongs to $\Pi_{n}$ if and only if there exists an $n$-fold computable sequence $\left\{ r_{\bm} \right\}_{\bm\in\NN^{n}} \subset \QN$ of rational numbers such that
\begin{equation*}
	x = \inf_{m_{1} \in \NN} \sup_{m_{2} \in \NN} \inf_{m_{3} \in \NN} \dots \underset{\bm \in \NN}\Theta r_{\bm}
\end{equation*}
and where $\Theta$ stands for the supremum if $n$ is even and for the infimum if $n$ is odd.
\item $\Delta_{n} = \Sigma_{n} \cap \Pi_{n}$.
\end{itemize}
\end{definition}

\begin{remark}
For every $n\geq 1$, one has $\Delta_{n} \subsetneq \Sigma_{n}$, $\Delta_{n} \subsetneq \Pi_{n}$, and $\Delta_{n} \subseteq \Delta_{n+1}$.
Recall that $\Delta_{1}$ is the set of all computable numbers, whereas $\Delta_{2}$ is the set of all recursively approximable numbers.
So the nested sets $\Delta_{n}$, $n\geq 1$ describe sets of numbers that are increasingly less computable as $n$ increases.
\end{remark}

\begin{remark}
It might be useful to note that for any $n\in\NN$, the set $\Delta_{n}$ is actually a field \cite{ZhengWeihrauch_MathLog01}, i.e. it is closed under the usual arithmetical operations addition, subtraction, multiplication and division.
In particular, $\Delta_{1}$ is the field of computable numbers.
\end{remark}

% =====================================================================================================
% ===== BIBLIOGRAPHIE =================================================================================
% =====================================================================================================
%\bibliographystyle{elsarticle-num} 
%\bibliography{../publications,../pub_books,../pub_Pohl}

\end{document}